\documentclass{llncs}
\pdfoutput=1
\usepackage{url}

\usepackage{hyperref}

\usepackage{amssymb}
\usepackage{latexsym}
\usepackage{amsmath}
\usepackage{amsthm}
\usepackage{xcolor}

\usepackage{algpseudocode}
\usepackage{algorithm}

\newtheorem{observation}{Observation}

\newtheorem*{ESS}{Equal Sum Subsets problem (ESS)}

\newtheorem*{SSR}{Subset-Sums Ratio problem (SSR)}
\newtheorem*{eSSR}{Subset-Sums Ratio problem (SSR) (equivalent definition)}

\newtheorem*{RestrSSR}{Restricted Subset-Sums Ratio problem}

\newtheorem*{MinMaxSSR}{Semi-Restricted Subset-Sums Ratio problem}

\begin{document}
\mainmatter              
\title{A Faster FPTAS for the Subset-Sums Ratio Problem}
\titlerunning{FPTAS for the SSR}  
%
\author{Nikolaos Melissinos\inst{1} \and Aris Pagourtzis\inst{2}}
\authorrunning{Nikolaos Melissinos and Aris Pagourtzis} 
%
\tocauthor{Nikolaos Melissinos, Aris Pagourtzis}
\institute{
School of Applied Mathematical and Physical Sciences\\
National and Technical University of Athens\\
Polytechnioupoli, 15780 Zografou, Athens, Greece,\\
\email{ge05019@central.ntua.gr}
\and
School of Electrical and Computer Engineering\\
National and Technical University of Athens\\
Polytechnioupoli, 15780 Zografou, Athens, Greece,\\
\email{pagour@cs.ntua.gr}
}

\maketitle              

\begin{abstract}
The Subset-Sums Ratio problem (SSR) is an optimization problem in which, 
given a set of integers, the goal is to find two subsets such that the ratio 
of their sums is as close to 1 as possible. In this paper we develop a new FPTAS 
for the SSR problem which builds on techniques proposed in [D. Nanongkai, 
Simple FPTAS for the subset-sums ratio problem, Inf. Proc. Lett. 113 
(2013)]. One of the key improvements of our scheme is the use of a dynamic programming table in 
which one dimension represents the \emph{difference} of the sums of the two subsets. This idea, together with a careful 
choice of a scaling parameter, yields an
FPTAS that is several orders of magnitude faster than the best currently known scheme of [C. Bazgan, M. Santha, 
Z. Tuza, Efficient approximation algorithms for the Subset-Sums Equality problem, 
J. Comp. System Sci. 64 (2) (2002)].

\keywords{approximation scheme, subset-sums ratio, knapsack problems, 
combinatorial optimization}
\end{abstract}
\section{Introduction}
\label{intro}
We study the optimization version of the following NP-hard decision problem which given a set of
integers asks for two subsets of equal sum (but, in contrast to the Partition problem, the 
two subsets do not have to form a partition of the given set):
\begin{ESS}
Given a set $A = \{a_1,\ldots,a_n\}$ of $n$ positive integers, are there two 
nonempty and disjoint sets $ S_1$, $S_2$ $\subseteq \{1,\ldots,n\}$ such that 
$ \sum_{i \in S_1} a_i = \sum_{j \in S_2} a_j \mbox{?} $
\end{ESS}
\noindent
Our motivation to study the ESS problem and its optimization version comes from the fact that it is a
fundamental problem closely related to problems appearing in many scientific areas. Some examples are 
the Partial Digest problem, which comes from molecular biology (see~\cite{cie:eid:pen,cie:eid}), 
the problem of allocating individual goods (see~\cite{lip:mar:mos:sab}), tournament construction (see~\cite{khan}), and a variation of the 
Subset Sum problem, namely the Multiple 
Integrated Sets SSP, which finds applications in the field of cryptography (see~\cite{vol}).

The ESS problem has been proven NP-hard by Woeginger and Yu in \cite{wo:yu} 
and several of its variations have been proven NP-hard by Cieliebak
et al.\ in \cite{cie:eid:pag:sch:tr,cie:eid:pag,cie:eid:pag:sch}. 
The corresponding optimization problem is:
\begin{SSR}
Given a set $A = \{a_1,\ldots,a_n\}$ of $n$ positive integers, find two 
nonempty and disjoint sets $ S_1$, $S_2$ $\subseteq \{1,\ldots,n\}$ 
that minimize the ratio 
\begin{align*}
 &\frac{\max\{\sum_{i \in S_1} a_i, \sum_{j \in S_2} a_j\}}{\min\{\sum_{i \in S_1} a_i, \sum_{j \in S_2} a_j\}} \ \ .
\end{align*}
\end{SSR}
\noindent
The SSR problem was introduced by Woeginger and Yu~\cite{wo:yu}. 
In the same work they present an $1.324$ approximation 
algorithm which runs in $O(n \log  n)$ time. 
The SSR problem received its first FPTAS by Bazgan et al.\ in \cite{baz:san:tuz}, 
which approximates the optimal solution in time no less than $O( {n^5}/{\varepsilon^3})$; to the best
of our knowledge this is still the faster scheme proposed for SSR. A second, 
simpler but slower, FPTAS was proposed by Nanongkai in \cite{nanon}.

The FPTAS we present in this paper makes use of some ideas  
proposed in~\cite{nanon}, strengthened by certain key improvements that lead to a considerable acceleration: 
our algorithm approximates 
the optimal solution in $O( {n^4}/{\varepsilon})$ time, several orders of magnitude faster 
than the best currently known scheme of~\cite{baz:san:tuz}.

\section{Preliminaries}
\label{preliminaries}
We will first define two functions that will allow us to simplify several of the expressions that we will need
throughout the paper.

\begin{definition}[Ratio of two subsets]
Given a set $A = \{a_1,\ldots,a_n\}$ of $n$ positive integers and two 
sets $ S_1, S_2 \subseteq \{1,\ldots,n\}$ we define 
$\mathcal{R}(S_1,S_2,A)$ as follows:
\begin{align*}
 \mathcal{R}(S_1,S_2,A)= 
 \begin{cases}
  \frac{\sum_{i \in S_1} a_i}{\sum_{j \in S_2} a_j} &\mbox{ if } S_1\cup S_2 \neq \emptyset,\\
  +\infty &\mbox{ otherwise.}
 \end{cases}
\end{align*}
\end{definition}

\begin{definition}[Max ratio of two subsets]
Given a set $A = \{a_1,\ldots,a_n\}$ of $n$ positive integers and two 
sets $ S_1, S_2 \subseteq \{1,\ldots,n\}$ we define 
$\mathcal{MR}(S_1,S_2,A)$ as follows:
\begin{align*}
 \mathcal{MR}(S_1,S_2,A)= \max\{ \mathcal{R}(S_1,S_2,A),  
\mathcal{R}(S_2,S_1,A)  \}\enspace.
\end{align*}
\end{definition}

\noindent
Note that, in cases where at least one of the sets is empty, 
the Max Ratio function will return $\infty$. 
Using these functions, the SSR problem can be rephrased as shown below.
\begin{eSSR}
Given a set $A = \{a_1,\ldots,a_n\}$ of $n$ positive integers, find two 
disjoint sets $ S_1$, $S_2$ $\subseteq \{1,\ldots,n\}$ such that the value 
$\mathcal{MR}(S_1,S_2,A)$ is minimized. 
\end{eSSR}
\noindent
In addition, from now on, whenever we have a set $A=\{a_1,\ldots,a_n\}$ we will 
assume that $0 < a_1< a_2 < \ldots < a_n$ (clearly, if the input contains two equal numbers then
the problem has a trivial solution).

The FPTAS proposed by Nanonghai~\cite{nanon} approximates the SSR problem 
by solving a restricted version. 
\begin{RestrSSR}
Given a set $A = \{a_1,\ldots,a_n\}$ of $n$ positive integers and two 
integers $ 1\leq p < q \leq n$, find two disjoint sets 
$ S_1$, $S_2$ $\subseteq \{1,\ldots,n\}$ such that $\{\max S_1,\max S_2\}=\{p,q\}$ 
and the value $\mathcal{MR}(S_1,S_2,A)$ is minimized. 
\end{RestrSSR}
\noindent
Inspired by this idea, we define a less restricted version. 
The new problem requires one additional input integer, instead of two, which represents 
the smallest of the two maximum elements of the sought optimal solution. 
\begin{MinMaxSSR}
Given a set $A = \{a_1,\ldots,$ $a_n\}$ of $n$ positive integers and an 
integer $ 1\leq p < n$, find two disjoint sets $ S_1$, 
$S_2$ $\subseteq \{1,\ldots,n\}$ such that $\max S_1=p<\max$ $S_2$ 
and the value $\mathcal{MR}(S_1,S_2,A)$ is minimized. 
\end{MinMaxSSR}
\noindent
Let $A =\{a_1,\ldots,a_n\}$ be a set of $n$ positive integers and $p \in \{1,\ldots,n-1\}$.
Observe that, if $S_1^*$, $S_2^*$ is the optimal solution of SSR problem 
of instance $A$ and  $S_1^p$, $S_2^p$ the optimal solution 
of Semi-Restricted SSR problem of instance $A$, $p$ then:
\begin{equation*}
 \mathcal{MR}(S_1^*,S_2^*,A) = \min_{p \in \{1,\ldots,n-1\}}  \mathcal{MR}(S_1^p,S_2^p,A) \enspace . 
\end{equation*}
\noindent
Thus, we can find the optimal solution of SSR problem by solving 
the SSR Semi-Restricted SSR problem for all $p \in \{1,\ldots,n-1\}$.

\section{Pseudo-polynomial time algorithm for Semi-Restricted SSR problem}\label{pseudo}

Let the $A$, $p$ be an instance of the Semi-Restricted SSR problem
where $A =\{a_1,\ldots,a_n\}$ and $1\leq p<n$.
For solving the problem we have to check two cases for the maximum 
element of the optimal solution. Let $S_1^*$, $S_2^*$ be the 
optimal solution of this instance and $\max S_2^*=q$. 
We define $B=\{a_i\mid i >p ,  a_i < \sum_{j=1}^p a_j\}$ and
$C= \{a_i \mid a_i \geq \sum_{j=1}^p a_j\}$ from which we have that
either $a_{q} \in B$ or $a_{q} \in C$. Note that $A=\{a_1,\ldots,a_p\} \cup B \cup C$. 

\medskip \noindent
\textbf{Case 1 ($a_q \in C$).}
It is easy to see that if $a_q \in C$, then $a_q= \min C$ 
and the optimal solution will be $(S_1=\{1,\ldots,p\}, S_2= \{q\})$. 
We describe below a function that returns this pair of sets, thus computing the optimal solution
if Case 1 holds.
\begin{definition}[Case 1 solution]
Given a set $A = \{a_1,\ldots,a_n\}$ of $n$ positive integers and an 
integer $ 1\leq p < n$ we define the function $\mathcal{SOL}_1(A,p)$ as follows:
\begin{align*}
 \mathcal{SOL}_1(A,p)= 
 \begin{cases}
  (\{1,\ldots,p\}, \{\min C\}) &\mbox{ if } C \neq \emptyset\mbox{,}\\
  (\emptyset, \emptyset) &\mbox{ otherwise,}
 \end{cases}
\end{align*}
where $C= \{a_i \mid a_i > \sum_{j=1}^p a_j\}$.
\end{definition}
\noindent
\textbf{Case 2 ($a_q \in B$).}
This second case is not trivial.
Here, we define an integer $m = \max\{j\mid a_j \in  A \smallsetminus C\}$ and 
a matrix $T$, where $T[i,d]$, $0 \leq i \leq m, -2\cdot \sum_{k=1}^p a_k \leq d \leq \sum_{k=1}^p a_k$, 
is a quadruple to be defined below.
A cell $T[i,d]$ is nonempty if there exist two disjoint sets 
$S_1$, $S_2$ with sums $sum_1$, $sum_2$ such that $sum_1-sum_2=d$, 
$\max S_1=p$, and $S_1 \cup S_2 \subseteq \{1,\ldots,i\} \cup \{p\}$; 
if $i> p$, we require in addition that $p <\max S_2$.
In such a case, cell $T[i,d]$ consists of the two sets $S_1$, $S_2$, 
and two integers $\max (S_1\cup S_2)$ and $sum_1 + sum_2$.
A crucial point in our algorithm is that if there exist more than one pairs of sets 
which meet the required conditions, we keep the one that maximize the 
value $sum_1+sum_2$; for convenience, we make use of a 
function to check this property and select the appropriate sets. The algorithm for this case 
(Algorithm~\ref{sub2}) finally returns the pair $S_1$, $S_2$
which, among those that appear in some $T[m,d]\neq \vec{\emptyset}$, has the 
smallest ratio $\mathcal{MR}(S_1,S_2,A)$.

\begin{definition}[Larger total sum tuple selection]
Given two tuples $\vec{v_1}=(S_{1},S_{2},q,x)$ and $\vec{v_2}=(S_{1}',S_{2}',q',x')$ 
we define the function $\mathcal{LTST}(\vec{v_1},\vec{v_2})$ as follows:
\begin{align*}
 \mathcal{LTST}(\vec{v_1},\vec{v_2})= 
 \begin{cases}
  \vec{v_2} &\mbox{ if } \vec{v_1} = \vec{\emptyset} \mbox{ or } x'>x \mbox{,}\\
  \vec{v_1} &\mbox{ otherwise}\enspace .
 \end{cases}
\end{align*}
\end{definition}

\begin{algorithm}[H]
\caption{Case 2 solution [$\mathcal{SOL}_2(A,p)$ function]}
\label{sub2}
\begin{algorithmic}[1]
  \Require a strictly sorted set $A=\{a_1,\ldots,a_n\}$, $a_i\in \mathbb{Z}^+$, and an integer $p$, $1\leq p<n$.
  \Ensure the sets of an optimal solution for Case 2.
  \State $S'_1\leftarrow \emptyset$, $S'_2\leftarrow \emptyset$
  \State $Q\leftarrow \sum_{i=1}^p a_i$,  $m \leftarrow \max \{i \mid a_i < Q \}$
  \If{$m>p$}
    \ForAll{$i\in \{0,\ldots,m\}$, $d\in \{-2\cdot Q,\ldots, Q\}$}
      \State $T[i,d] \leftarrow \vec{\emptyset}$
    \EndFor
	\algstore{myalg}
\end{algorithmic}
\end{algorithm} 
	
\begin{algorithm}[H]
\begin{algorithmic}[1]
      \algrestore{myalg}
    \State $T[0,a_p]\leftarrow (\{p\},\emptyset,p,a_p)$ \Comment{$p\in S_1$ by problem definition}
    \For{$i\leftarrow1$ \textbf{to} $m$}
      \If{$i<p$}
	\ForAll{$T[i-1,d]  \neq \vec{\emptyset}$}
	  \State $(S_1,S_2,q,x)\leftarrow T[i-1,d] $
	  \State $T[i,d]\leftarrow \mathcal{LTST}(T[i,d], T[i-1,d])$ \label{line12}
	  \State $T[i,d+a_i]\leftarrow \mathcal{LTST}(T[i,d+a_i], (S_1\cup\{i\},S_2,q,x+a_i))$ \label{line13}
	  \State $T[i,d-a_i]\leftarrow \mathcal{LTST}(T[i,d-a_i], (S_1,S_2\cup\{i\},q,x+a_i))$ \label{line14}
	\EndFor
      \ElsIf{$i=p$} \Comment{$p$ is already placed in $S_1$}
	\ForAll{$T[i-1,d] \neq \vec{\emptyset}$}
	  \State $T[i,d]\leftarrow T[i-1,d]$
	\EndFor
      \Else
	\ForAll{$T[i-1,d] \neq \vec{\emptyset}$}    \label{line21}
	  \State $(S_1,S_2,q,x) \leftarrow T[i-1,d]$
	  \If{$i>p+1$} \label{line23}
	    \State $T[i,d]\leftarrow \mathcal{LTST}( T[i,d],  T[i-1,d])$
	  \EndIf \label{line25}
	  \If{$d-a_i\geq -2 \cdot Q$}
	    \State $T[i,d-a_i]\leftarrow \mathcal{LTST}( T[i,d-a_i], (S_1,S_2\cup\{i\},i,x+a_i))$  \label{line27}
	  \EndIf
	\EndFor
	\ForAll{$T[p,d]  \neq \vec{\emptyset}$} 
	  \State $(S_1,S_2,q,x)\leftarrow T[p,d] $ \label{line31}
	  \If{$d-a_i\geq -2\cdot Q$}
	    \State $T[i,d-a_i]\leftarrow \mathcal{LTST}( T[i,d-a_i],  (S_1,S_2\cup\{i\},i,x+a_i))$
	  \EndIf  \label{line34}
	\EndFor   \label{line35}
      \EndIf
    \EndFor
    \For{$d \leftarrow -2 \cdot Q$ \textbf{to} $ Q$ }
      \State $(S_1,S_2,q,x)\leftarrow T[m,d] $
      \If{$\mathcal{MR}(S_1,S_2,A)< \mathcal{MR}(S'_1,S'_2,A)$}
	\State $S'_1\leftarrow S_1$, $S'_2\leftarrow S_2$
      \EndIf
    \EndFor
  \EndIf
  \State \Return $S'_1$, $S'_2$
\end{algorithmic}
\end{algorithm}

\noindent
We next present the complete algorithm for Semi-Restricted SSR (Algorithm~\ref{Alg1}) which simply returns the best 
among the two solutions obtained by solving the two cases.
Algorithm~\ref{Alg1} runs in time polynomial in $n$ and $Q$ 
(where $Q = \sum_{i=1}^p a_i$), therefore  
it is a pseudo-polynomial time algorithm. More precisely, by using appropriate data 
structures we can store the sets in the matrix cells in $O(1)$ time (and space) per cell, 
which implies that the time complexity of the algorithm is $O(n \cdot Q)$.

\begin{algorithm}[H]
\caption{Exact solution for Semi-Restricted SSR [$\mathcal{SOL}_{ex}(A,p)$ function]}
\label{Alg1}
\begin{algorithmic}[1]
  \Require a strictly sorted set $A=\{a_1,\ldots,a_n\}$, $a_i\in \mathbb{Z}^+$, and an integer $p$, $1\leq p<n$.
  \Ensure the sets of an optimal solution of Semi-Restricted SSR.
  \State $(S_1,S_2) \leftarrow \mathcal{SOL}_1(A,p)$
  \State $(S'_1,S'_2) \leftarrow \mathcal{SOL}_2(A,p)$
  \If{$\mathcal{MR}( S_1, S_2,A) \leq \mathcal{MR}(S'_1,S'_2,A) $} 
    \State \Return $S_1$, $S_2$
  \Else
    \State \Return $S'_1$, $S'_2$
  \EndIf
\end{algorithmic}
\end{algorithm}  
\noindent

\section{Correctness of the Semi-Restricted SSR algorithm}
\label{correctness}
In this section we will prove that Algorithm~\ref{Alg1} solves 
exactly the Semi-Restricted SSR problem.
Let $S_1^*$, $S_2^*$ be the sets of an optimal solution for 
input $(A=\{a_1,\ldots,a_n\},p)$.   

Starting with the case 1 (where $\max S_2^* \in \{i\mid a_i \geq \sum_{j=1}^p a_j\} $), is easy to see that:
\begin{observation}
The sets $S_1^*=\{1,\ldots,p\}$, $S_2^*=\{ \min\{i\mid a_i \geq \sum_{j=1}^p a_j\} \}$ 
give the optimal ratio. 
\end{observation}
\noindent
Those are the sets which the function $\mathcal{SOL}_1(A,p)$ returns.

For the case 2 (where $\max S_2^* \in \{i\mid i>p, a_i < \sum_{j=1}^p a_j\} $) 
we have to show that the cell $T[m,d]$ (where $d= \sum_{i \in S_1^*}a_i -\sum_{j \in S_2^*}a_j$) 
contains two sets $S_1$, $S_2$ with ratio equal to optimum. Before that we will show a lemma 
for the sums of the sets of the optimal solution.
\begin{lemma}\label{lemma_sums}
Let $ Q= \sum_{i=1}^p a_i$ then we have $\sum_{i \in S_1^*}a_i \leq Q$ and $\sum_{i \in S_2^*}a_i  < 2\cdot Q$.
\end{lemma}

\begin{proof}
Observe that $\max  S_1^*=p$. This gives us $\sum_{i \in S_1^*}a_i \leq \sum_{i=1}^p a_i$
so it remains to prove $ \sum_{i \in S_2^*}a_i  < 2\cdot Q $.
Suppose that $\sum_{i \in S_2^*}a_i  \geq 2\cdot Q $.
We can define the set $S_2$ as $S_2^* \smallsetminus \{\min  S_2^*\}$.
Note that, for all $i \in S_2^*$, we have that the $a_i < \sum_{i=1}^p a_i$.
Because of that,
\begin{align*}
\sum_{i \in S_1^*}a_i  \leq \sum_{i=1}^p a_i < \sum_{i \in S_{2}}a_i < \sum_{i \in S_2^*}a_i
\end{align*}
which means that the pair $(S_1^*, S_{2})$ is a feasible solution with smaller max ratio 
than the optimal, which is a contradiction.
\end{proof}
\noindent
The next two lemmas describe same conditions which guarantee that 
the cells of $T$ are nonempty. Furthermore, they secure that we will 
store the appropriate sets to return an optimal solution.
\begin{lemma}\label{cells_1}
If there exist two disjoint sets $(S_1, S_2)$ such that
 \begin{itemize}
  \item $\max S_2 < \max S_1=p$
  \item $\sum_{i \in S_1}a_i - \sum_{j \in S_2}a_j =d$
 \end{itemize}
then $T[i,d] \neq \emptyset$ for all $p\geq i \geq \max (S_1 \cup S_2\smallsetminus \{p\})$.
Furthermore for the sets  $(S'_1, S'_2)$ which are stored in $T[i,d]$ it holds that
$$ \sum_{i \in S'_1}a_i + \sum_{j \in S'_2}a_j \geq \sum_{i \in S_1}a_i + \sum_{j \in S_2}a_j \enspace .  $$
\end{lemma}

\begin{proof}
Note that, for all pairs $(S_1, S_2)$ which meet the conditions, their sums 
are smaller than $Q$ because $\max  (S_1\cup S_2)=p$ so for the value 
$ d=\sum_{i \in S_1}a_i - \sum_{j \in S_2}a_j$ we have 
$$-Q\leq d\leq Q \enspace.$$
The same clearly holds for every pair of subsets of $S_1$, $S_2$.

We will prove the lemma by induction on $q= \max (S_1 \cup S_2\smallsetminus \{p\})$. 
For convenience if $S_1 \cup S_2\smallsetminus \{p\} = \emptyset$ we let $q=0$.\\
$\bullet$ $q=0$ (base case).\\ 
The only pair which meets the conditions for $q=0$ is the 
$(\{p\},\emptyset)$.
Observe that cell $T[0,a_p]$ is nonempty by the construction of the table 
and the same holds for $T[i,a_p]$, $1\leq i\leq p$ (by line~\ref{line12}).
In this case the pair of sets which meets the conditions and the pair which is stored 
are exactly the same, so the lemma statement is obviously true.\\
$\bullet$ Assume that the lemma statement holds for $q= k\leq p-1$; we will prove it for $q=k+1$ as well.\\
Let $(S_1, S_2)$ be a pair of sets which meets the conditions. 
Either $q\in S_1$ or $q\in S_2$; therefore either $(S_1\smallsetminus\{q\}, S_2)$ or $(S_1, S_2\smallsetminus\{q\})$ 
(respectively) meets the conditions.
By the inductive hypothesis, we know that 
 \begin{itemize}
  \item either $T[q-1,d-a_q]$ or $T[q-1,d+a_q]$ (resp.) is nonempty 
  \item in any of the above cases for the stored pair $(S'_1,S'_2)$ it holds that: \\
  $ \sum_{i \in S'_1}a_i + \sum_{j \in S'_2}a_j  \geq \sum_{i \in S_1}a_i + \sum_{j \in S_2}a_j-a_q$
 \end{itemize} 
In particular, if $(S_1\smallsetminus\{q\},S_2)$ meets the 
conditions then $T[q-1,d-a_q]$ is nonempty. In line~\ref{line13}  
$q$ is added to the first set and therefore $T[q,d]$ is nonempty and the 
stored pair is $(S'_1\cup \{q\}, S'_2)$ (or some other with larger total sum). Hence, 
the total sum of the pair in $T[q,d]$ is at least
$$ \sum_{i \in S'_1}a_i + \sum_{j \in S'_2}a_j +a_q  \geq \sum_{i \in S_1}a_i + \sum_{j \in S_2}a_j \enspace . $$
If on the other hand $(S_1, S_2\smallsetminus\{q\})$ is the pair that meets the  
conditions then $T[q-1,d+a_q]$ is nonempty. In line~\ref{line14} 
$q$ is added to the second set and therefore $T[q,d]$ is nonempty and the 
stored pair is $(S'_1, S'_2\cup \{q\})$ (or other with larger total sum). Hence, the total sum 
of the pair in $T[q,d]$ is at least
$$ \sum_{i \in S'_1}a_i + \sum_{j \in S'_2}a_j +a_q  \geq \sum_{i \in S_1}a_i + \sum_{j \in S_2}a_j \enspace . $$
The same holds for cells $T[i,d]$ with $q<i\leq p$ (due to line~\ref{line12}).\\
This concludes the proof. 
\end{proof}

\noindent
A similar lemma can be proved for sets with maximum element index greater than $p$.

\begin{lemma}\label{cells_2}
If there exist two disjoint sets $(S_1, S_2)$ such that
 \begin{itemize}
  \item $\max S_2 =q  > p = \max S_1$
  \item $\sum_{i \in S_1}a_i \leq Q$, $\sum_{j \in S_2}a_j <2\cdot Q$
  \item $\sum_{i \in S_1}a_i - \sum_{j \in S_2}a_j =d$
 \end{itemize}
then $T[i,d] \neq \emptyset$ for all $ i \geq q$.
Furthermore for the sets  $(S'_1, S'_2)$ which are stored in $T[i,d]$ it holds that
$$ \sum_{i \in S'_1}a_i + \sum_{j \in S'_2}a_j \geq \sum_{i \in S_1}a_i + \sum_{j \in S_2}a_j \enspace .  $$
\end{lemma}

\begin{proof}
Note that, for all pairs $(S_1, S_2)$ which meet the conditions, the value 
$ d=\sum_{i \in S_1}a_i - \sum_{j \in S_2}a_j$ it holds that
$$-2\cdot Q\leq d \leq Q \enspace . $$
The same clearly holds for every pair of subsets of $S_1$, $S_2$.

We will prove the lemma by induction.  
Let $(S_1, S_2)$ meet the conditions and $q=\max S_2$.\\
$\bullet$ $q=p+1$ (base case)\\
Clearly $\max S_2=p+1$ so the sets $(S_1,S_2\smallsetminus\{p+1\})$ meet the conditions
of the Lemma~\ref{cells_1} which gives us that 
 \begin{itemize}
  \item $T[p,d+a_{p+1}]$ is nonempty
  \item for the stored pair $(S'_1, S'_2)$ it holds that: \\
  $ \sum_{i \in S'_1}a_i + \sum_{j \in S'_2}a_j \geq \sum_{i \in S_1}a_i + \sum_{j \in S_2}a_j-a_{p+1}$
 \end{itemize}
Having the $T[p,d+a_{p+1}] \neq \vec{\emptyset}$ the algorithm uses it in lines~\ref{line31}-\ref{line34} 
and adds $p+1$ to the second (stored) set so, we have that $T[p+1,d]$ is nonempty and the stored 
sets have total sum (at least):
$$  \sum_{i \in S'_1}a_i + \sum_{j \in S'_2}a_j +a_{p+1} \geq \sum_{i \in S_1}a_i + \sum_{j \in S_2}a_j \enspace. $$
Furthermore, because $T[p+1,d]$ is nonempty the above hold, additionally, for all $T[i,d]$, $i> p+1$ (because the 
condition at line~\ref{line23} is met, the algorithm fills those cells).
The above conclude the base case.\\
$\bullet$ Assuming that the lemma statement holds for $q= k>p$, we will prove it for $q=k+1$.\\
Here we have to check two cases. Either $\max(S_2 \smallsetminus \{q\})>p$ or not.
\smallskip

\noindent
\textbf{Case 1} ($\max(S_2 \smallsetminus \{q\})>p$)\textbf{.} 
The pair of sets $(S_1, S_2\smallsetminus\{q\})$ meets the conditions; 
by the inductive hypothesis, we have
\begin{itemize}
  \item $T[q-1,d-a_{q}]$ is nonempty 
  \item for the stored pair $(S'_1, S'_2)$ it holds that: \\
  $ \sum_{i \in S'_1}a_i + \sum_{j \in S'_2}a_j \geq \sum_{i \in S_1}a_i + \sum_{j \in S_2}a_j-a_{q}$
\end{itemize}
Having the $T[q-1,d+a_{p+1}] \neq \vec{\emptyset}$ the algorithm uses it in line~\ref{line27} 
and adds $q$ to the second (stored) set so we have that $T[q,d]$ is nonempty and the stored 
sets have total sum (at least):
$$  \sum_{i \in S'_1}a_i + \sum_{j \in S'_2}a_j +a_{q} \geq \sum_{i \in S_1}a_i + \sum_{j \in S_2}a_j \enspace. $$

As before, the same holds for the cells  $T[i,d]$ with $i> p+1$ because the condition 
at line~\ref{line23} is met.
\smallskip

\noindent
\textbf{Case 2} ($\max(S_2 \smallsetminus \{q\})<p$)\textbf{.} 
The sets $(S_1, S_2\smallsetminus\{q\})$ meets the conditions of the Lemma~\ref{cells_1} (because $\max S_1 = p$)
which gives that
 
 \begin{itemize}
  \item $T[p,d+a_{q}]$ is nonempty
  \item for the stored pair $(S'_1, S'_2)$ it holds that: \\
  $ \sum_{i \in S'_1}a_i + \sum_{j \in S'_2}a_j \geq \sum_{i \in S_1}a_i + \sum_{j \in S_2}a_j-a_{q}$
 \end{itemize}
Having $T[p,d+a_{q}] \neq \vec{\emptyset}$ the algorithm uses it in lines~\ref{line31}-\ref{line34}
and adds $q$ to the second (stored) set so we have that $T[q,d]$ is nonempty and the stored 
sets have total sum (at least):
$$  \sum_{i \in S'_1}a_i + \sum_{j \in S'_2}a_j +a_{q} \geq \sum_{i \in S_1}a_i + \sum_{j \in S_2}a_j \enspace. $$
Furthermore, because $T[q,d]$ is nonempty the previous hold for all $T[i,d]$, $i>q\geq p+1$ (because the condition 
at line~\ref{line23} is met).
\end{proof}

\noindent
Now we can prove that, in the second case, the pair of sets which the algorithm 
returns and the pair of sets of an optimal solution have the same ratio. 
\begin{lemma}
If $(S'_1,S'_2)$ is the pair of sets that Algorithm~\ref{sub2} returns, then: 
$$\mathcal{MR}(S'_1,S'_2,A)=\mathcal{MR}(S_1^*, S_2^*,A) \enspace.$$
\end{lemma}

\begin{proof}
Let $m$ be the size of the first dimension of the matrix $T$. 
Observe that for all $ i$, $p+1\leq i \leq m$, the sets $S_1$, $S_2$  of the nonempty cells $T[i,d]$ 
are constructed (lines~\ref{line21}-\ref{line35} of Algorithm~\ref{sub2}) such that $\max S_1=p$ and $i \geq \max S_2 >p$. 
Therefore the pair $(S'_1,S'_2)$ returned by the algorithm is a feasible solution.
We can see that the sets $S_1^*$, $S_2^*$ meet the conditions of Lemma~\ref{cells_2} (the conditions 
for the sums are met because of Lemma~\ref{lemma_sums})
which give us that the cell $T[m,d]$ (where $d=\sum_{i \in S_1^*}a_i - \sum_{j \in S_2^* }a_j$) 
is non empty and contains two sets with total sum non less than $\sum_{i \in S_1^*}a_i + \sum_{j \in S_2^* }a_j$.
Let $S_1$, $S_2$ be the sets which are stored to the cell $T[m,d]$. Then we have
\begin{align}
 \mathcal{MR}(S'_1,S'_2,A)\leq \mathcal{MR}(S_1,S_2,A)\leq \mathcal{MR}(S_1^*, S_2^*,A) \label{opt_case_2}
\end{align}
where the second inequality is because
$$ \sum_{i \in S_1^*}a_i - \sum_{j \in S_2^* }a_j =\sum_{i \in S_1}a_i - \sum_{j \in S_2 }a_j$$
and
$$ \sum_{i \in S_1^*}a_i + \sum_{j \in S_2^* }a_j \leq \sum_{i \in S_1}a_i + \sum_{j \in S_2 }a_j \enspace . $$
By the Eq.\ref{opt_case_2} and because the $S_1^*$, $S_2^*$ have the smallest Max Ratio we have
$$\mathcal{MR}(S'_1, S'_2,A) = \mathcal{MR}(S_1^*, S_2^*,A) \enspace.$$
\end{proof}

Now, we can write the next theorem, which follows by the previous cases.

\begin{theorem}
Algorithm~\ref{Alg1} returns an optimal solution for Semi-Restricted SSR.
\end{theorem}

\section{FPTAS for Semi-Restricted SSR and SSR}
\label{FPTAS}

Algorithm~\ref{Alg1}, which we presented at Section~\ref{pseudo}, is an exact 
pseudo-polynomial time algorithm for the Semi-Restricted SSR problem. In order 
to derivee a $(1+\varepsilon)$-approximation algorithm we will define a scaling 
parameter $\delta=\frac{\varepsilon \cdot a_p}{3 \cdot n}$ which we will use to 
make a new set $A'=\{a'_1,\ldots,a'_n\}$ with $a'_i=\lfloor \frac{a_i}{\delta}\rfloor$.
The approximation algorithm solves the problem optimally on input $(A',p)$ and returns the 
sets of this exact solution. The ratio of those sets is a $(1+ \varepsilon)$-approximation 
of the optimal ratio of the original input.

\begin{algorithm}[H]
\caption{FPTAS for Semi-Restricted SSR [$\mathcal{SOL}_{apx}(A,p,\varepsilon)$ function]}
\label{Alg2}
\begin{algorithmic}[1]
  \Require a strictly sorted set $A=\{a_1,\ldots,a_n\}$, $a_i\in \mathbb{Z}^+$, an integer $p$, $1\leq p<n$, 
  and an error parameter $\varepsilon\in (0,1)$.
  \Ensure the sets of a $(1+ \varepsilon)$-approximation solution for Semi-Restricted SSR.
  \State $\delta \leftarrow \frac{\varepsilon \cdot a_p}{3\cdot n}$
  \State $A'\leftarrow \emptyset$
  \For{$i\leftarrow1$ to $n$}
    \State $a'_i \leftarrow \lfloor \frac{a_i}{\delta} \rfloor $
    \State $A'\leftarrow A'\cup \{ a'_i \}$
  \EndFor
  \State $(S_1,S_2) \leftarrow \mathcal{SOL}_{ex}(A',p)$
  \State \Return $S_1$, $S_2$
\end{algorithmic}
\end{algorithm}  
\noindent 
Now, we will prove that the algorithm approximates the optimal 
solution by factor $(1+ \varepsilon)$. Our proof follows closely 
the proof of Theorem 2 in~\cite{nanon}. 

Let $S_{A}$, $S_{B}$ be the pair of sets returned by 
Algorithm~\ref{Alg2} on input $A=\{a_1,\ldots,a_n\}$, $p$ and 
$\varepsilon$ and $(S_1^*,S_2^*)$ be an optimal solution 
to the problem.
\begin{lemma} \label{proof_lemma_1}
For any $S \in \{ S_{A},S_{B}, S_1^*,  S_2^*\}$
\begin{align}
  \sum_{i \in S} a_i - n\cdot \delta &\leq  
 \sum_{i \in S} \delta \cdot a'_i \leq \sum_{i \in S} a_i,
 \label{scale1}                                             \\
  n\cdot \delta &\leq \frac{\varepsilon}{3} \cdot \sum_{i \in S} a_i.
 \label{scale2}
\end{align}
\end{lemma}

\begin{proof}
For Eq.~\eqref{scale1} notice that for all  
$i \in \{1,\ldots,n\} $ we define $a'_i =\lfloor 
\frac{a_i}{\delta} \rfloor$. This gives us
\begin{align*}
 \frac{a_i}{\delta} -1 \leq a'_i \leq  \frac{a_i}{\delta}  \Rightarrow 
 a_i -\delta \leq \delta \cdot a_i \leq a_i. 
\end{align*} 
In addition, for any 
$S \in \{ S_{A},S_{B}, S_1^*,  S_2^*\}$ we have 
$|S|\leq n$, which means that
$$  \sum_{i \in S} a_i - n\cdot \delta \leq  
 \sum_{i \in S} \delta \cdot a'_i \leq \sum_{i \in S} a_i.$$
For the Eq.~\eqref{scale2} observe that $\max  S \geq p$ for any 
$S \in \{ S_{A},S_{B}, S_1^*,  S_2^*\}$. By this observation, we can show 
the second inequality

$$
 n\cdot \delta \leq  n \cdot \frac{\varepsilon \cdot a_p}{3\cdot n} 
\leq  \frac{\varepsilon}{3} \cdot \sum_{i \in S }a_i. \qed
$$
\end{proof}
\begin{lemma}\label{proof_lemma_2}
$
 \mathcal{MR}(S_{A},S_{B},A) \leq  
\mathcal{MR}(S_{A},S_{B},A')+ \frac{\varepsilon}{3}
$
\end{lemma}

\begin{proof}
\begin{align*}
\mathcal{R}(S_{A},S_{B},A)  = \frac{\sum_{i \in S_{A}}a_i}{\sum_{j \in S_{B}}a_j}
				   & \leq \frac{\sum_{i \in S_{A}}\delta\cdot a'_i + \delta \cdot n}{\sum_{j \in S_{B}}a_j} 
				   &\quad \mbox{[by Eq.~} \eqref{scale1}\mbox{]} \\
				   & \leq \frac{\sum_{i \in S_{A}} a'_i }{\sum_{j \in S_{B}}a'_j} +  
				   \frac{ \delta \cdot n}{\sum_{j \in S_{B}}a_j}
				   &\quad \mbox{[by Eq.~} \eqref{scale1}\mbox{]} \\
				   & \leq \mathcal{MR}(S_{A},S_{B},A')+ \frac{\varepsilon}{3} 
				   &\quad \mbox{[by Eq.~} \eqref{scale2}\mbox{]}
\end{align*}
The same way, we have 
$$\mathcal{R}(S_{B},S_{A},A) \leq \mathcal{MR}(S_{A},S_{B},A')+ \frac{\varepsilon}{3} $$
thus the lemma holds.
\end{proof}
\begin{lemma}\label{proof_lemma_3}
For any $\varepsilon\in (0,1)$, $ \mathcal{MR}(S_1^*,S_2^*,A') \leq (1+\frac{\varepsilon}{2})\cdot \mathcal{MR}(S_1^*,S_2^*,A)$.
\end{lemma}

\begin{proof}
If $\mathcal{R}(S_1^*,S_2^*,A')\geq 1$, let  $(S_1,S_2)=(S_1^*,S_2^*)$, 
otherwise $(S_1,S_2)=(S_2^*,$ $S_1^*)$. That gives us
\begin{align*}
 \mathcal{MR}(S_1^*,S_2^*,A') = \mathcal{R}&(S_{1},S_{2},A')  = \frac{\sum_{i \in S_{1}}a'_i}{\sum_{j \in S_{2}}a'_j} &\\
 &  \leq \frac{\sum_{i \in S_{1}}a_i}{\sum_{j \in S_{2}}a_j - n\cdot \delta} & \mbox{[by Eq.~}\eqref{scale1}\mbox{]} \\
 & = \frac{\sum_{i \in S_{2}}a_i}{\sum_{j \in S_{2}}a_j - n\cdot \delta} \cdot \frac{\sum_{i \in S_{1}}a_i}{\sum_{j \in S_{2}}a_j}& \\
 & = (1+ \frac{n\cdot \delta}{\sum_{j \in S_{2}}a_j - n\cdot \delta} ) \cdot \frac{\sum_{i \in S_{1}}a_i}{\sum_{j \in S_{2}}a_j}\enspace.&
\end{align*}
Because $S_2 \in \{S_1^*, S_2^*\}$ by Eq.~\eqref{scale2} it  follows that 
\begin{align*}
 \mathcal{MR}(S_1^*,S_2^*,A')& \leq  (1 + \frac{1}{ \frac{3}{\varepsilon }- 1})\cdot \frac{\sum_{i \in S_{1}}a_i}{\sum_{j \in S_{2}}a_j}& \\
 & =(1 + \frac{\varepsilon}{ 3- \varepsilon})\cdot \frac{\sum_{i \in S_{1}}a_i}{\sum_{j \in S_{2}}a_j}& \\
 & \leq (1 + \frac{\varepsilon}{ 2})\cdot \frac{\sum_{i \in S_{1}}a_i}{\sum_{j \in S_{2}}a_j}&  \mbox{[because }\varepsilon \in (0,1)\mbox{]}\\
 & \leq (1 + \frac{\varepsilon}{ 2})\cdot \mathcal{MR}(S_1^*,S_2^*,A).&
\end{align*}
This concludes the proof.
\end{proof}

Now we can prove that Algorithm~\ref{Alg2} is a $(1+\varepsilon)$ approximation algorithm.

\begin{theorem}\label{fptas2}
Let $S_{A}$, $S_{B}$ be the pair of sets returned by Algorithm~\ref{Alg2} 
on input ($A=\{a_1,\ldots,a_n\}$, $p$, $\varepsilon$) and $S_1^*$, $S_2^*$ be an optimal solution, then:
$$ \mathcal{MR}(S_{A}, S_{B},A) 
\leq (1+\varepsilon)\cdot \mathcal{MR}(S_1^*, S_2^*,A). $$
\end{theorem}

\begin{proof}
The theorem follows from a sequence of inequalities:
\begin{align*}
 \mathcal{MR}(S_{B},S_{A},A) 
 &\leq \mathcal{MR}(S_{A},S_{B},A')+ \frac{\varepsilon}{3}& \mbox{[by Lemma }\ref{proof_lemma_2}\mbox{]} \\
 & \leq \mathcal{MR}(S_1^*,S_2^*,A')+ \frac{\varepsilon}{3}&\\
 & \leq (1+\frac{\varepsilon}{2})\cdot \mathcal{MR}(S_1^*,S_2^*,A) + \frac{\varepsilon}{3}& \mbox{[by Lemma }\ref{proof_lemma_3}\mbox{]}\\
 & \leq (1+\varepsilon)\cdot \mathcal{MR}(S_1^*,S_2^*,A).&
\end{align*}
\end{proof}

It remains to show that the complexity of Algorithm~\ref{Alg2} 
is $\mathcal{O}(poly(n,{1}/{\varepsilon}))$. Like we said at \ref{pseudo} the 
algorithm solves the Semi-Restricted SSR problem in $O(n \cdot Q)$ 
(where $Q=\sum_{i=1}^p a'_i$). We have to bound the value of $Q$.
By the definition of $a'_i$ we have, 
$$ Q=\sum_{i=1}^p a'_i \leq n\cdot a'_p\leq \frac{n\cdot a_p}{\delta}= \frac{3\cdot n^2}{\varepsilon}$$
which means that Algorithm~\ref{Alg2} runs in $O({n^3}/{\varepsilon})$.
\medskip

\noindent
Clearly, it suffices to perform $n-1$ executions of the FPTAS for Semi-Restricted SSR (Algorithm~\ref{Alg2}), and pick the best of the returned solutions, in order to obtain an FPTAS for the (unrestricted) SSR problem. Therefore, 
we obtain the following.

\begin{theorem}
The above described algorithm is an FPTAS for SSR that runs in $O({n^4}/{\varepsilon})$ time. 
\end{theorem}

\section{Conclusion}
\label{conclusion}

In this paper we provide an FPTAS for the Subset-Sums Ratio (SSR) problem that is much faster than the 
best currently known scheme of Bazgan et al.~\cite{baz:san:tuz}. 
There are two novel ideas that provide this improvement. 
The first comes from observing that in~\cite{nanon} the proof of correctness essentially relies 
only on the value of the smallest 
of the two maximum elements; this led to the idea to use only that information in order to solve the problem 
by defining and solving a new variation which we call Semi-Restricted SSR.
A key ingredient in our approximation scheme is the use, in the scaling parameter $\delta$, of 
a value smaller than the sums of the sets of both optimal and approximate solutions 
(which in our case is the value of the smallest of the two maximum elements). We believe that this technique can
be used in several other partition problems, e.g.\ such as those described in~\cite{lip:mar:mos:sab,vol}.

The second idea was to use one dimension only, for the difference of the sums of the two sets, instead of 
two dimensions, one for each sum. This idea, combined with the observation that between 
two pairs of sets with the same difference, the one with the largest total sum has
ratio closer to~1, is the key to obtain an optimal solution in much less time.
It's interesting to see whether and how this technique could be used to problems that seek more 
than two subsets.

A natural open question 
is whether our techniques can be applied to obtain approximation results for other variations of 
the SSR problem \cite{cie:eid:pag,cie:eid:pag:sch}.

%
%

\end{document}